\documentclass[runningheads]{llncs}

\usepackage[utf8]{inputenc}
\usepackage{amsmath}
\usepackage{amsfonts}
\usepackage{amssymb}
\usepackage{lmodern}
\usepackage{mathtools}
\usepackage{subcaption}

\usepackage{graphicx}
\usepackage{braket}
\usepackage{comment}
\usepackage{xfrac}

\usepackage{booktabs}
\usepackage{adjustbox}

\usepackage{qcircuit}

\usepackage{hyperref, xcolor}

\usepackage{pgf}
\usepackage{float}

\usepackage{todonotes}
\usepackage{caption} 
\newcommand{\cent}[0]{\mbox{\textcent}}
\newcommand{\dollar}{\$}
\newcommand{\mymatrix}[2]{\left( \begin{array}{#1} #2 \end{array} \right)}
\newcommand{\myvector}[1]{\mymatrix{c}{#1}}
\newcommand{\mypar}[1]{\left( #1 \right)}

\newcommand{\abs}[1]{\left\lvert#1\right\rvert}

\newcommand{\setBuilder}[2] { 
    \ensuremath{ \{{#1} \mid {#2}\}}
}

\newcommand{\MOD}[1]{\ensuremath{\mathtt{MOD}_{\rm #1}}}
\newcommand{\MODp}[0]{\MOD{p}}

\newcommand{\bigO}{\ensuremath{\mathcal{O}}}

\newcommand{\scl}[1]{\scalebox{0.9}{#1}}

\newcommand{\thh}[0]{\textsuperscript{th}}

\usepackage{url}

\begin{document}

\title{Cost-efficient QFA Algorithm for Quantum Computers\thanks{Part of this research was done while Yakary{\i}lmaz was visiting the Institute of Theoretical and Applied Informatics, Gliwice, Poland in 2021 and 2022.}}

%
%

\author{\"Ozlem Salehi\inst{1,3} \and Abuzer Yakary{\i}lmaz\inst{2,3}
}


%
\authorrunning{\"O. Salehi and A. Yakary{\i}lmaz}
%
\institute{Institute of Theoretical and Applied Informatics, Polish Academy of Sciences, Gliwice, Poland \and
Center for Quantum Computer Science, University of Latvia, R\={\i}ga, Latvia \and
QWorld Association, Tallinn, Estonia, \url{https://qworld.net} \\
\email{osalehi@iitis.pl,abuzer.yakaryilmaz@lu.lv}
}

\maketitle  

\begin{abstract}
The study of quantum finite automata (QFAs) is one of the possible approaches in exploring quantum computers with finite memory. Despite being one of the most restricted models, Moore-Crutchfield quantum finite automaton (MCQFA) is proven to be exponentially more succinct than classical finite automata models in recognizing certain languages such as \(\MODp = \setBuilder{a^{j}}{j \equiv 0 \mod p}\), where $p$ is a prime number. In this paper, we present a modified MCQFA algorithm for the language $\MODp$, the operators of which are selected based on the basis gates on the available real quantum computers. As a consequence, we obtain shorter quantum programs using fewer basis gates compared to the implementation of the original algorithm given in the literature.

\end{abstract}

\section{Introduction}

Quantum finite automaton (QFA) is a theoretical model for studying quantum computers with finite memory (finite number of states). The most natural extension from a classical automaton to QFA is obtained by replacing the transition matrices with unitary operators. The obtained model is called Moore-Crutchfield quantum finite automaton (MCQFA) \cite{MC00}, and it is one of the earliest models \cite{AY21}. Even though some regular languages can not be recognized by MCQFAs, MCQFAs are proven to be exponentially more succinct (memory-efficient or state-efficient) for certain languages\footnote{MCQFAs can also be super-exponentially more succinct but on promise problems \cite{AY12}.} such as the language \(\MODp = \setBuilder{a^{j}}{j \equiv 0 \mod p}\), where $p$ is a prime number \cite{AF98}. 

The proposed MCQFA recognizing $\MODp$ language consists of sub-automata whose actions can be visualized as a sequence of rotations on the 2D plane, correspondingly around the $y$-axis on the Bloch Sphere. Several implementation ideas have been proposed in \cite{Kalis18} to realize the algorithms in gate-based quantum computers. Some new ideas and simulation results on IBMQ backends \cite{IBM} have recently appeared in \cite{BSONY21}. (See also \cite{TFL19,MPC20,PHYF22} for photonics or optical implementations of QFA algorithms.)

When a quantum circuit is run on an IBMQ computer, the circuit is restructured to fit the topology of the computer, arbitrary gates are decomposed into basis gates, and the circuit is optimized to improve the performance. This process is known as transpilation. As the depth and the number of 2-qubit gates increase, the performance deteriorates due to noise. Hence, it becomes crucial to design quantum circuits with smaller depths and fewer gates.

With this motivation, we propose a modified MCQFA algorithm for the $\MODp$ language. The proposed algorithm performs rotations around the $x$-axis on the Bloch Sphere and its correctness is proven through a change of basis. We show that the optimized implementation of the original algorithm from \cite{Kalis18} can be adapted for our modified algorithm. When transpiled to be run on the real quantum backends, our algorithm is more efficient in the number of gates and has a smaller circuit depth. We also provide some experimental results from $\mathit{ibmq\_belem}$ backend.

The rest of the paper is structured as follows: In Section 2, we give the necessary background information. In Section 3, we present both the original and modified versions of 2-state and $\log$-state MCQFAs. In Section 4, we discuss the implementation of the algorithms and we conclude with Section 5.

\section{Preliminaries}

In this section, we provide information on quantum finite automata.

We denote the finite alphabet by $\Sigma$ not containing the left end-marker $\cent$ and the right end-marker $\$$, and \(\tilde{\Sigma}\) denotes \(\Sigma \cup \{\cent, \$ \}\). The length of any string \(w\in \Sigma^{*}\) is denoted by \(\abs{w}\), and $ w[i] $ denotes its $ i $\thh symbol. The string $w$ is processed by an automaton as $\cent w \dollar$ from left to right and symbol by symbol.

Among the many different quantum finite automaton (QFA) models proposed in the literature \cite{AY21}, we will focus on \emph{Moore-Crutchfield quantum finite automaton} (MCQFA) \cite{MC00}, which is known as the most restricted model. 

Formally, a MCQFA $ M $ with \(d\) states is a 5-tuple \[
    M = (Q, \Sigma,  \{U_\sigma \mid \sigma \in \tilde{\Sigma}\}, q_s, Q_A),
\] where \(Q=\{q_1,\ldots,q_d\}\) is the \emph{set of states}, \(U_\sigma\) is the \emph{unitary operator} for symbol \(\sigma \in \tilde{\Sigma}\), \(q_s \in Q\) is the \emph{start state}, and \(Q_A \subseteq Q\) is the \emph{set of accepting state(s)}.

We can trace the computation of $ M $ on a given input $ w \in \Sigma^* $ with length $l$ by a $d$-dimensional column vector (quantum state or state vector), where the $j$\textsuperscript{th} entry represents the amplitude of $ q_j $. At the beginning of computation, $ M $ starts in quantum state $ \ket{v_0} = \ket{q_s} $, which has zeros except its $ s $\textsuperscript{th} entry that is 1. For each symbol $ \sigma \in \tilde{\Sigma} $, the unitary operator $ U_\sigma $ is applied to the state vector. Thus, $ M $ is in the following final state after reading the whole input:
\begin{equation}\label{eq:state}
    \ket{v_f} = U_{\dollar} U_{w[l]} U_{w[l-1]} \cdots U_{w[1]} U_{\cent} \ket{v_0}.  
\end{equation}
Then, the final state is measured in the computational basis, and the input is accepted if and only if an accepting state is observed. The accepting probability is calculated as in Eq.~\eqref{eq:prob}
\begin{equation}\label{eq:prob}
    \sum_{q_j \in Q_A} | \braket{q_j|v_f} |^2,
\end{equation}
where $  \braket{q_j|v_f} $ gives the amplitude of $ q_j $ in $ \ket{v_f} $.

In Table \ref{tabl:notation}, we present a list of notations used throughout the paper.

\begin{table}[]
\centering
\resizebox{\textwidth}{!}{%
\begin{tabular}{@{}ll@{}}
\toprule
\textbf{Notation}                         & \textbf{Explanation}                                                                      \\ \midrule
$Q$, $\tilde{Q}$                          & Set of states for 2-state and $\bigO(\log p)$-state MCQFAs                                \\ \midrule
$q_s$, $\tilde{q}_s$                      & Start state for 2-state and $\bigO(\log p)$-state MCQFAs                                  \\ \midrule
$Q_A$, $\tilde{Q}_A$                      & Set of accept states for 2-state and $\bigO(\log p)$-state MCQFAs                         \\ \midrule
$\Sigma$, $\tilde{\Sigma}$            & Set of alphabet, set of alphabet including end-markers                                    \\ \midrule
$U_{\sigma}$, ${U}_a'$                    & Unitary operators for symbol $\sigma$ for 2-state MCQFAs (original and new)               \\ \midrule
$\tilde{U}_\sigma$, $\tilde{U}_\sigma'$ & Unitary operators for symbol $\sigma$ for $\bigO(\log p)$-state MCQFAs (original and new) \\ \midrule
$\hat{U}_a$, $\hat{U}_a'$                 & Unitary operators for symbol $\sigma$ used in the optimized implementation (original and new) \\ \midrule
$w$, $w[i]$                               & String, $i$'th symbol of string $w$                                                         \\ \midrule
$M_p$, ${M}_p'$,                          & 2 state MCQFAs with rotation angle $2 \pi / p$ (old, new)                                 \\ \midrule
$M_p^k$                                   & 2 state MCQFA with rotation angle $2 \pi k / p$                                           \\ \midrule
$K$                                       & A subset of integers from $ \{0,\dots, p-1\}$                                                     \\ \midrule
$M_p^K$, $(\tilde{M}_p^{K})'$             & $\bigO(\log p)$-state MCQFA with rotation angles $2 \pi k_j /p$, $k_j \in K$ (old, new)       \\ \midrule
$SX$                                      & Matrix representation corresponding to the squareroot of NOT operator                     \\ \midrule
$R_j$                                     &  Rotation on the real-valued plane spanned by $\ket{0}$ and $\ket{1}$ with angle $2 \pi k_j /p$                                \\ \midrule
$R_j'$                                    & Rotation matrix with angle $4 \pi k_j /p$ on the $x$-$y$ plane of the Bloch Sphere                                                                                          \\ \midrule
$V$, $V'$                                 & Vector spaces                                                                             \\ \midrule
$CX, I, S_x, X$                           & Controlled NOT, identity, squareroot of NOT, NOT gates                                    \\ \midrule
$R_y, R_z$                                & Rotation gates around $y$ and $z$ axis                                                    \\ \midrule
$F(\sigma,\rho)$                          & Fidelity of the states represented by density matrices $\sigma$ and $\rho$                                                \\ \bottomrule
\end{tabular}%
}
\caption{List of notations used throughout the paper. }
\label{tabl:notation}
\end{table}

\section{MCQFAs for the $\MODp$ Problem}

Any bounded-error probabilistic finite automaton recognizing $\MODp$ requires at least $p$ states while it was proven in \cite{AF98} that there exists an MCQFA with \(\bigO(\log p)\) states recognizing the language $\MODp$ with bounded error, i.e., any input is accepted with probability 1 if it is in $ \MODp $, and it is accepted with probability at most $ \varepsilon $ for some $ \varepsilon < \frac{1}{2} $, otherwise\footnote{The number of states depending on the error bound $ \varepsilon $ is bounded by $ \frac{4}{\varepsilon} \log 2 p  $ \cite{AN09}.}.

\subsection{2-state MCQFA for $\MODp$}

Let us start by describing the well-known 2-state MCQFA  construction for the language $\MODp$. Let $M_p = (\{q_1,q_2\},\Sigma = \{a\}, \{U_\sigma \mid \sigma \in \tilde{\Sigma}\}, q_1, \{q_1\}),$ where the unitary operators $U_{\cent}$ and $U_{\$}$ are defined as the $2 \times 2$ identity matrix, and operator $U_a$ is defined as 
$
    \begin{pmatrix}
        \cos{(2\pi/p)} & ~ & -\sin{(2\pi/p)} \\
        \sin{(2\pi/p)} & & \cos{(2\pi/p)} \\
    \end{pmatrix},
$
and corresponds to a counter-clockwise rotation with angle \({2\pi}/{p}\) on the 2D plane spanned by $\ket{q_1}$ and $\ket{q_2}$ as visualized in Figure \ref{fig:2d}.

\begin{figure}
    \centering

\tikzset{every picture/.style={line width=0.75pt}} 

\begin{tikzpicture}[x=0.75pt,y=0.75pt,yscale=-0.7,xscale=0.7]

\draw    (200.51,30.33) -- (201.49,226.33) ;
\draw [shift={(201.5,229.33)}, rotate = 269.72] [fill={rgb, 255:red, 0; green, 0; blue, 0 }  ][line width=0.08]  [draw opacity=0] (8.93,-4.29) -- (0,0) -- (8.93,4.29) -- cycle    ;
\draw [shift={(200.5,27.33)}, rotate = 89.72] [fill={rgb, 255:red, 0; green, 0; blue, 0 }  ][line width=0.08]  [draw opacity=0] (8.93,-4.29) -- (0,0) -- (8.93,4.29) -- cycle    ;
\draw    (104,128.33) -- (298,128.33) ;
\draw [shift={(301,128.33)}, rotate = 180] [fill={rgb, 255:red, 0; green, 0; blue, 0 }  ][line width=0.08]  [draw opacity=0] (8.93,-4.29) -- (0,0) -- (8.93,4.29) -- cycle    ;
\draw [shift={(101,128.33)}, rotate = 0] [fill={rgb, 255:red, 0; green, 0; blue, 0 }  ][line width=0.08]  [draw opacity=0] (8.93,-4.29) -- (0,0) -- (8.93,4.29) -- cycle    ;
\draw   (118.33,128.33) .. controls (118.33,82.68) and (155.34,45.67) .. (201,45.67) .. controls (246.66,45.67) and (283.67,82.68) .. (283.67,128.33) .. controls (283.67,173.99) and (246.66,211) .. (201,211) .. controls (155.34,211) and (118.33,173.99) .. (118.33,128.33) -- cycle ;
\draw [line width=1.5]    (201,128.33) -- (259.52,76) ;
\draw [shift={(262.5,73.33)}, rotate = 500.19] [fill={rgb, 255:red, 0; green, 0; blue, 0 }  ][line width=0.08]  [draw opacity=0] (11.61,-5.58) -- (0,0) -- (11.61,5.58) -- cycle    ;
\draw [line width=0.75]    (229.34,104.14) .. controls (248.34,100.14) and (250.34,120.14) .. (241.34,129.14) ;
\draw    (242.62,105.92) -- (243.45,113.87) ;
\draw    (242.62,105.92) -- (250.18,109.62) ;

\draw (308,119.4) node [anchor=north west][inner sep=0.75pt]    {$|q_{1} \rangle $};
\draw (185,5.4) node [anchor=north west][inner sep=0.75pt]    {$|q_{2} \rangle $};
\draw (57,119.4) node [anchor=north west][inner sep=0.75pt]    {$-|q_{1} \rangle $};
\draw (179,235.4) node [anchor=north west][inner sep=0.75pt]    {$-|q_{2} \rangle $};
\draw (236,90.05) node [anchor=north west][inner sep=0.75pt]    {$2\pi/p$};
\end{tikzpicture}
    \caption{Visualization of the operator $U_a$ on the 2D plane.}
    \label{fig:2d}
\end{figure}
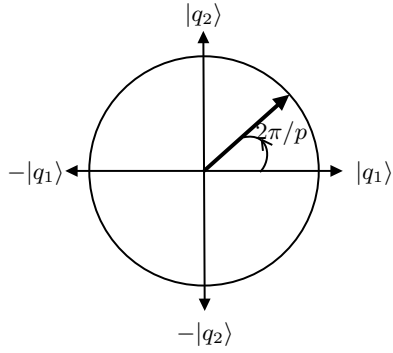

After reading the string $a^j$, the state of $M_p$ is $\cos (2\pi j/p)  \ket{q_1} + \sin (2\pi j/p) \ket{q_2}$. It is easy to see that for member strings, that is when $j$ is a multiple of $p$, the probability of acceptance is equal to 1. For non-member strings,  the acceptance probability is given by $ \cos^2 \mypar{{2\pi}  \abs{w}/{p}} $, and it takes its maximum value when $ |w| \equiv \frac{p-1}{2} \mod p $ or $ |w| \equiv \frac{p+1}{2} \mod p $. Thus, the error for non-member strings can be as big as $ 1 - \cos^2(\pi |w| /p)  $.

Next, we will describe an alternative MCQFA for the language $\MODp$, based on the following observation. Let us define ${U}_a'$ and $SX$ as in Eq.~\eqref{eq:matrices}.
\begin{equation}\label{eq:matrices}
    {U}_a' =
    \begin{pmatrix}
        e^{-2 \pi i /p} & ~ & 0 \\
        0 & &e^{2 \pi i /p} 
    \end{pmatrix}
   \text{ and } 
SX = \frac{1}{2} \begin{pmatrix}
        1+i & ~ & 1-i \\
        1-i &~ & 1+i \\
    \end{pmatrix}
\end{equation}

\begin{lemma} \label{lem: sxua}
$SX^{\dagger} U_a' SX = U_a$.
   \end{lemma}
\begin{proof}
To start with note that 
\begin{equation}
SX = \frac{1}{2} \begin{pmatrix}
        1+i & ~ & 1-i \\
        1-i &~ & 1+i \\
    \end{pmatrix}
    =
    \frac{1}{\sqrt{2}}
    \begin{pmatrix}
        e^{i \pi / 4} & ~ & e^{-i\pi / 4} \\
        e^{-i\pi /4} &~ & e^{i \pi /4} \\
    \end{pmatrix}
    =
    \frac{1}{\sqrt{2}}
    e^{-i \pi /4}
    \begin{pmatrix}
        e^{i \pi /2} & ~ & 1 \\
        1 &~ & e^{i \pi /2} \\
    \end{pmatrix}.
\end{equation}
\noindent 
We calculate $ SX^{\dagger} U_a' SX$
\begin{align}
    & =
    \frac{1}{\sqrt{2}} e^{i \pi / 4} \mymatrix{cc}{ e^{-i \pi /2} & 1 \\ 
    1 &  e^{-i \pi /2}
    }
    \mymatrix{cc}{ e^{-2\pi i / p } & 0 \\ 0 & e^{ 2\pi i /p} }
    \frac{1}{\sqrt{2}} e^{-i \pi / 4}
    \begin{pmatrix}
        e^{i\pi / 2} & ~ & 1 \\
        1 &~ & e^{i \pi / 2} \\
    \end{pmatrix}  \nonumber
    \\    & 
    = \frac{1}{2}e^0 
    \mymatrix{cc}{ 
    e^{-i \pi /2} & 1 \\
    1 &  e^{-i \pi /2}
    }
    \mymatrix{cc}{
    e^{ i \mypar{ {\pi / 2}- {2\pi / p} } } & 
    e^{ -i 2\pi / p } \\
    e^{ i 2\pi / p } &
    e^{ i \mypar{ {\pi / 2}+{2 \pi /p} } }
    } \nonumber 
    \\ & 
    = \frac{1}{2}
    \mymatrix{ccc}{
        e^{ -i 2\pi / p } + e^{ i 2\pi /p }  & ~~~~ &
        e^{ -i \mypar{ {\pi /2} + {2\pi / p}  } } +
        e^{ i \mypar{ {\pi / 2} + {2\pi /p}  } } \\
        e^{ i \mypar{ {\pi /2} - {2\pi /p}  } } +
        e^{ -i \mypar{ {\pi /2} - {2\pi /p}  } } & ~~~~ &
        e^{ -i 2\pi / p } + e^{ i 2\pi / p }
    }.
\end{align}

We calculate each term of this matrix one by one by using the following trigonometric qualities: (i) $ \cos(-\theta) = \cos(\theta) $, (ii) $ \sin(-\theta) = -\sin(\theta) $, (iii) $ e^{i\theta} + e^{-i\theta} = 2\cos(\theta) $, and (iv) $ \cos (\theta+\pi/2) = -sin(\theta) $.

The top-left and the bottom-right terms are equal to
\begin{equation}
    e^{ -i 2\pi / p } + e^{ i 2\pi / p } = 2 \cos \mypar{ 2\pi / p }.
\end{equation}
The top-right term is equal to
\begin{align}
     e^{ -i \mypar{ {\pi / 2} + {2\pi/p}  } } +
    e^{ i \mypar{ {\pi / 2} + {2\pi / p}  } } 
    &= 
    2\cos \mypar{ {\pi / 2} + {2\pi / p} }  \nonumber \\ 
    &= -2 \sin \mypar{ 2\pi / p } .
\end{align}
The bottom-left term is equal to 
\begin{align}
    e^{ i \mypar{ {\pi/2} - {2\pi/p}  } } +
        e^{ -i \mypar{ {\pi / 2} - {2\pi / p}  } } &=  2\cos \mypar{ {\pi/2} - {2\pi/p}} \nonumber \\ 
        &= -2 \sin \mypar{ - 2\pi / p } \nonumber \\ 
        &= 2 \sin \mypar{ 2\pi / p }.
\end{align}
Thus, we obtain that
$
    SX^{\dagger} U_a' SX  = 
    \mymatrix{ccc}{
     cos(2\pi / p) & ~~ & -sin(2\pi / p) \\
     sin(2\pi / p) & ~~ & cos(2\pi / p)
     }
     = U_a.
$
\qed
\end{proof}

The above lemma suggests that the operator $U_a$ for symbol $a$ can be replaced with $SX^{\dagger} U_a' SX$. Now note that $(SX^{\dagger} U_a' SX)^j = SX^{\dagger} U_a'^j SX$ for any $j\geq 0$. This suggests a new automaton that applies $SX$ at the beginning of the computation, applies the unitary operator $U_a'$ for each scanned $a$, and applies $SX^{\dagger}$ at the end of the computation. Based on this observation, we describe the following MCQFA.

Let ${M}_p' = (Q, \Sigma, \{{U}_{\sigma}' \mid \sigma \in \tilde{\Sigma}\}, q_s, Q_A) $ where $\Sigma$, $Q$, $q_s$ and $Q_A$ are defined as before. The operator ${U}_{\cent}'$ is defined as in Eq.~\eqref{eq:ucent}.

\begin{equation} \label{eq:ucent}
 {U}_{\cent}' = SX= \frac{1}{2} 
\begin{pmatrix}
        1+i & ~ & 1-i \\
        1-i &~ & 1+i \\
    \end{pmatrix}
\end{equation}
The operator ${U}_a'$ is defined as in Eq.~\eqref{eq:ua}.

\begin{equation} \label{eq:ua}
    {U}_a' =
    \begin{pmatrix}
        e^{-2 \pi i /p} & ~ & 0 \\
        0 & &e^{2 \pi i /p} 
    \end{pmatrix}
\end{equation}
$\tilde{U}_{\$} = \tilde{U}_{\cent}^{\dagger}$ and it is represented with the following matrix:

\begin{equation}
{U}_{\$} =  SX^{\dagger} = \frac{1}{2} 
\begin{pmatrix}
        1-i & ~ & 1+i \\
        1+i &~ & 1-i \\
    \end{pmatrix}.
\end{equation}

The relationship between $M_p'$ and $M_p$ can be also explained with a change of basis between $V'=\{\ket{0'}, \ket{1'} \}$ and $V=\{\ket{0},\ket{1}\}$ where
\begin{equation}
    \ket{0'} = 
    \frac{1}{2}\myvector{1+i\\1-i}~~\mbox{and}~~
    \ket{1'} = \frac{1}{2} \myvector{1-i\\1+i},
\end{equation}
and with the change of basis matrix $SX^{\dagger}$ that maps $\ket{0'}$ to $\ket{0}$ and $\ket{1'}$ to $\ket{1}$. In this manner, the linear transformation $U_a'$ working on the vector space spanned by $V'$ is mapped to the vector space spanned by $V$ through the mapping  $SX^{\dagger}U_a'SX$.

\subsection{{\(\bigO(log p)\)}-State MCQFA for $\MODp$}

In the above constructions, the rotation angle is selected as $2\pi/p$. Note that one can use the rotation angle \(2 \pi k/p\) for some \(k \in \{1,\ldots,p-1\}\), and depending on the value of $k$, the acceptance probability for each non-member string varies. Nevertheless, the maximum acceptance probability of a non-member string is still not bounded. Combining multiple automata with different rotation angles, it is proven that for any $p$, there exists a $\bigO( \log p )$-state MCQFA recognizing the $\MODp$ language with bounded error in \cite{AF98}. Let us start by recalling the construction.

2-state MCQFA $ M_p^k $ is defined similar to $M_p$, except that it performs a rotation of angle $2\pi k/{p}$ on the 2D plane. Let $K = \{ k_1,\ldots,k_d \} $ where each $ k_j \in \{1,\ldots,p-1\}$. We define $2d$-state MCQFA $ M_p^K  = (\tilde{Q},\Sigma, \{\tilde{U}_\sigma \mid \sigma \in \tilde{\Sigma}\}, \tilde{q}_s, \tilde{Q}_A)  $. The state set $\tilde{Q}$ is defined as $\{ q_1^1,q_2^1,q_1^2,q_2^2,\ldots,q_1^d,q_2^d \}$, where $q_1^j, q_2^j$ belongs to state set of $M_p^{k_j}$,  $\tilde{q}_s= q_1^1$ and $\tilde{Q}_A = \{ q_1^1\}$.

$ M_p^K $ simulates the individual $M_p^{k_j}$'s in parallel, through the following unitary transformations. $ \tilde{U}_{\cent} $ is defined as
$
H^{\otimes \log(2d) -1} \otimes I
$
and it creates the equal superposition state given in Eq.~\eqref{eq:sp}.
\begin{align}\label{eq:sp}
    \ket{q_1^1} \; \xrightarrow{\makebox[8mm]{\(\tilde{U}_{\cent}\)}} \; \frac{1}{\sqrt{d}} \ket{q_1^1} + \frac{1}{\sqrt{d}} \ket{q_1^2}+\cdots + \frac{1}{\sqrt{d}} \ket{q_1^d}
\end{align}

Upon reading each $a$, the operator $\tilde{U}_a$ is applied, which is defined as in Eq.~\eqref{eq:tildeua}.

\begin{align}\label{eq:tildeua}
\tilde{U}_a = \bigoplus_{j=1}^d R_j  = \mymatrix{c c c c}{
        R_1 & 0 & \cdots & 0 \\
        0 & R_2 & \cdots & 0 \\
        \vdots & \vdots &  \ddots & \vdots \\
        0 & 0 & \cdots & R_d
    },
\end{align}

where
\begin{align}
    R_{j} =
    \mymatrix{c@{\hspace{3mm}}c} {
    \cos{\left({2\pi k_j}/{p}\right)} & -\sin{\left({2\pi k_j}/{p}\right)} \\
    \sin{\left({2\pi k_j}/{p}\right)} & \cos{\left({2\pi k_j}/{p}\right)} 
    }.
\end{align}
The operator $\tilde{U}_a$ helps executing each \(M_p^{k_j}\), where \(M_p^{k_j}\)s are counter-clockwise rotations by angle \(2 \pi k_j/p\). The unitary operator $\tilde{U}_{\$}$ is defined as $\tilde{U}_{\cent}^{-1}$. 

In \cite{AN09}, it is proven that after reading $a^l$, the acceptance probability is equal to 1 if $l$ is divisible by $p$ and otherwise equals
$\frac{1}{d^2} \sum_{j=1}^d \cos^2 (2\pi k_j l/p)$; and, when $j$ is not a multiple of $p$, the acceptance probability is less than $\varepsilon$ for $d = 2 \frac{\log p}{\varepsilon}$, resulting in a \(\bigO(\log p)\)-state MCQFA recognizing the $\MODp$ language with bounded error.

Next we will describe an alternative \(\bigO(\log p)\)-state MCQFA using the observation made in Lemma \ref{lem: sxua}. To start with, note that by Lemma \ref{lem: sxua}, it follows that $SX^{\dagger}R_j' SX = R_j \text{ for any } j=1,\dots, d$, where
\begin{align}
 R_{j}' &=
    \begin{pmatrix}
        e^{-2 \pi i k_j /p} & ~ & 0 \\
        0 & &e^{2 \pi i k_j /p} 
    \end{pmatrix} =
    e^{-2 \pi i k_j /p}
    \begin{pmatrix}
        1 & ~ & 0 \\
        0 & &e^{4 \pi i k_j /p} 
    \end{pmatrix}.
\end{align}

\noindent Let's define $\tilde{U}_a'$ as follows:
\begin{equation}
\tilde{U}_a' = \bigoplus_{j=1}^d R_j'  = \mymatrix{c c c c}{
        R_1' & 0 & \cdots & 0 \\
        0 & R_2' & \cdots & 0 \\
        \vdots & \vdots &  \ddots & \vdots \\
        0 & 0 & \cdots & R'_d
    }.
\end{equation}
Since $ SX^{\dagger}SX=I$, the following holds:

\begin{align}
&\phantom{=}\mymatrix{c c c c}{
        SX^{\dagger}R_1'SX & 0 & \cdots & 0 \\
        0 & SX^{\dagger}R_2'SX & \cdots & 0 \\
        \vdots & \vdots &  \ddots & \vdots \\
        0 & 0 & \cdots & SX^{\dagger}R_d'SX
    }^j  \nonumber \\
    &= 
\mymatrix{c c c c}{
        SX^{\dagger} & 0 & \cdots & 0 \\
        0 & SX^{\dagger}& \cdots & 0 \\
        \vdots & \vdots &  \ddots & \vdots \\
        0 & 0 & \cdots & SX^{\dagger}
    }
\mymatrix{c c c c}{
        R_1' & 0 & \cdots & 0 \\
        0 & R_2' & \cdots & 0 \\
        \vdots & \vdots &  \ddots & \vdots \\
        0 & 0 & \cdots & R_d'
    }^j
\mymatrix{c c c c}{
        SX & 0 & \cdots & 0 \\
        0 & SX & \cdots & 0 \\
        \vdots & \vdots &  \ddots & \vdots \\
        0 & 0 & \cdots & SX
    }  \nonumber \\
    &= (I^{\otimes \log (2d)-1} \otimes SX^{\dagger})  (\tilde{U}_a')^j (I^{\otimes \log (2d)-1} \otimes SX)
\end{align}

Using the above equality, we present automaton $ (\tilde{M}_p^{K})'  = (\tilde{Q},\Sigma,  \{\tilde{U}_{\sigma}' \mid \sigma \in \tilde{\Sigma}\}, \tilde{q}_s, \tilde{Q}_A)$ where $\Sigma$, $\tilde{Q}$, $\tilde{q}_s$ and $\tilde{Q}_A = \{q_1^1 \}$ are defined as above. The unitary operator $\tilde{U}_{\cent}'$ is defined as in Eq.~\eqref{eq:tildeucent},
\begin{align}
    \label{eq:tildeucent}
    \tilde{U}_{\cent}' = H^{\otimes \log (2d) -1 } \otimes U_{\cent}',
\end{align}
the operator $\tilde{U}_a'$ is defined as above, and $\tilde{U}_{\$}'$ is defined as in Eq.~\eqref{eq:tildeudollar},
\begin{align}
    \label{eq:tildeudollar}
    \tilde{U}_{\$}' = \tilde{U}_{\cent}'^{\dagger} = H^{\otimes \log (2d) -1 } \otimes {U}_{\$}'.
\end{align}

\section{Implementations}

In this section, we will present the quantum circuit implementations for the automata given in the previous section.

We used Qiskit quantum programming framework provided by IBM Quantum \cite{IBM}, which allows designing quantum circuits that can be simulated in local simulators as well as on real quantum computers. Before running a circuit in a real machine, the circuit is transpiled into basis gates defined as the set $\{CX, I, R_z, S_x, X \}$. The matrices corresponding to the gates in the basis set are given in Figure \ref{fig: basis}. The transpilation process also depends on the chosen backend and the selected optimization level. To run our experiments, we used $\mathit{ibmq\_belem}$ backend and optimization level 2.  

\begin{figure}
\centering 
\includegraphics[width = 0.8\textwidth]{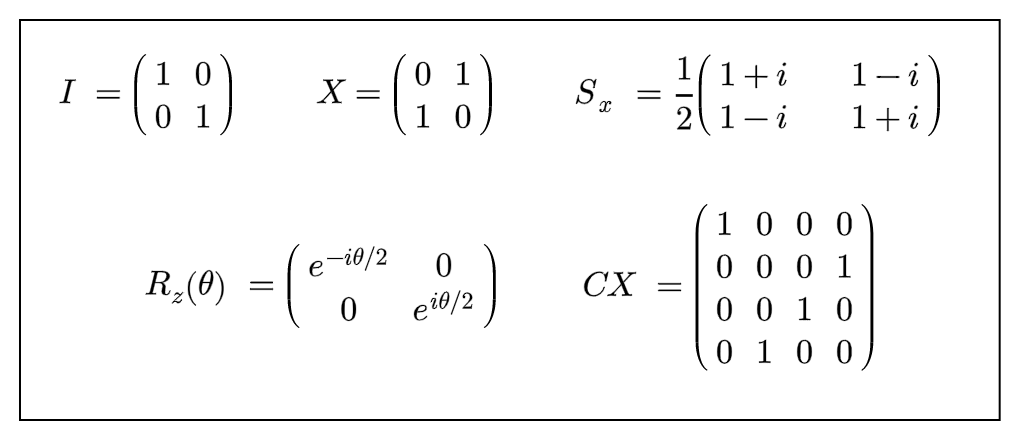}
\caption{Basis gates and their matrix representations in Qiskit.}
\label{fig: basis}
\end{figure}

\subsection{Single Qubit Implementation} 

Let us recall the implementation of single-qubit MCQFA ${M}_{11}$ from \cite{BSONY21}. The circuit diagram for the string $aa$ is given in Figure 
\ref{circ:4-1_single}. The operator $U_a$ is simply implemented by a $R_y$ gate with the angle of rotation $4\pi/11$, since in the Bloch sphere the angle between the states $\ket{0}$ and $\ket{1}$ is 180 degrees ($\pi$ radians). The matrix corresponding to $R_y$ gate in Qiskit is defined as in Eq.\eqref{eq:rygate}.
\begin{equation}\label{eq:rygate}
    R_y(2\theta) = \mymatrix{ccc}{ \cos \theta & ~ &  -\sin \theta \\ \sin \theta & & \cos \theta}.
\end{equation}

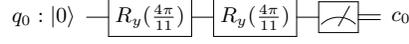
\begin{figure}[h]

\vspace{-0.1in}
    \centering
    \begin{equation*}
        \scl{\(
        \Qcircuit @C=1.0em @R=1em{
            \lstick{q_0:\ket0}  & \gate{R_y(\frac{4 \pi }{11})} & \gate{R_y(\frac{4 \pi }{11})} &   \meter & \rstick{c_0} \cw 
        }\)}
    \end{equation*}
    \caption{Single qubit \MOD{11} implementation for the string $aa$. The gate $R_y(\frac{4\pi}{11})$ is applied upon reading each $a$. }
    \label{circ:4-1_single}
\end{figure}
In Figure \ref{fig: rybloch}, the computation is visualized on the Bloch sphere for a string of length 11 and 4, respectively. The computation starts in state $\ket{0}$ which is shown by the green vector and expressed as
\begin{equation}
    \ket{v_0} = \myvector{1 \\ 0 }.
\end{equation}
In the figure on the left-hand side, the states obtained after applying $j$ $R_y$ gates are numbered as $j$ and shown using blue dots. After 11 rotations, the final quantum state is $\ket{0}$, resulting in the acceptance of the string with probability 1. The state vectors are presented in Eq.~\eqref{eq:states}.
\scriptsize
\begin{equation} \label{eq:states}
\begin{array}{l}
     \ket{v_0} = \myvector{1 \\ \\ 0} 
    \underbrace{\xrightarrow{R_y \mypar{ \frac{4\pi}{11} }}}_{1\textsuperscript{st}~a}
    \myvector{ \cos \frac{2\pi}{11} \\ \\ \sin \frac{2\pi}{11}}
    \underbrace{\xrightarrow{R_y \mypar{ \frac{4\pi}{11} }}}_{2\textsuperscript{nd}~a}
    \myvector{ \cos \frac{4\pi}{11} \\ \\ \sin \frac{4\pi}{11}}
    \underbrace{\xrightarrow{R_y \mypar{ \frac{4\pi}{11} }}}_{3\textsuperscript{rd}~a}
    \myvector{ \cos \frac{6\pi}{11} \\ \\ \sin \frac{6\pi}{11}}
    \underbrace{\xrightarrow{R_y \mypar{ \frac{4\pi}{11} }}}_{4\textsuperscript{th}~a}
\\ \\
    \myvector{ \cos \frac{8\pi}{11} \\ \\ \sin \frac{8\pi}{11}}
    \underbrace{\xrightarrow{R_y \mypar{ \frac{4\pi}{11} }}}_{5\textsuperscript{th}~a}
    \myvector{ \cos \frac{10\pi}{11} \\ \\ \sin \frac{10\pi}{11}}
    \underbrace{\xrightarrow{R_y \mypar{ \frac{6\pi}{11} }}}_{5\textsuperscript{th}~a}
    \myvector{ \cos \frac{12\pi}{11} \\ \\ \sin \frac{12\pi}{11}}
    \underbrace{\xrightarrow{R_y \mypar{ \frac{4\pi}{11} }}}_{7\textsuperscript{th}~a}
    \myvector{ \cos \frac{14\pi}{11} \\ \\ \sin \frac{14\pi}{11}}
    \underbrace{\xrightarrow{R_y \mypar{ \frac{4\pi}{11} }}}_{8\textsuperscript{th}~a}
\\ \\
    \myvector{ \cos \frac{16\pi}{11} \\ \\ \sin \frac{16\pi}{11}}
    \underbrace{\xrightarrow{R_y \mypar{ \frac{4\pi}{11} }}}_{9\textsuperscript{th}~a}
    \myvector{ \cos \frac{18\pi}{11} \\ \\ \sin \frac{18\pi}{11}}
    \underbrace{\xrightarrow{R_y \mypar{ \frac{4\pi}{11} }}}_{10\textsuperscript{th}~a}
    \myvector{ \cos \frac{20\pi}{11} \\ \\ \sin \frac{20\pi}{11}}
    \underbrace{\xrightarrow{R_y \mypar{ \frac{4\pi}{11} }}}_{11\textsuperscript{th}~a}
    \myvector{ \cos \frac{22\pi}{11} \\ \\ \sin \frac{22\pi}{11}} = 
    \myvector{1 \\ \\ 0}
\end{array}    
\end{equation}
\normalsize
On the right hand side, $R_y$ gate is applied for 4 times and the final vector is shown in red. All states lie on the $x$-$z$ plane. Remark that each $ R_y \mypar{ \frac{4\pi}{11} } $ rotates the state on the $x$-$z$ plane of the Bloch sphere with angle $ \frac{4\pi}{11} $, which corresponds the rotation with angle $ \frac{2\pi}{11} $ on $\ket{0}$-$\ket{1}$ plane, as shown in Figure~\ref{fig:2d}. 
\begin{figure}
    \centering
    \centering
\begin{subfigure}{.3\textwidth}
  \centering
  \includegraphics[width=1\linewidth]{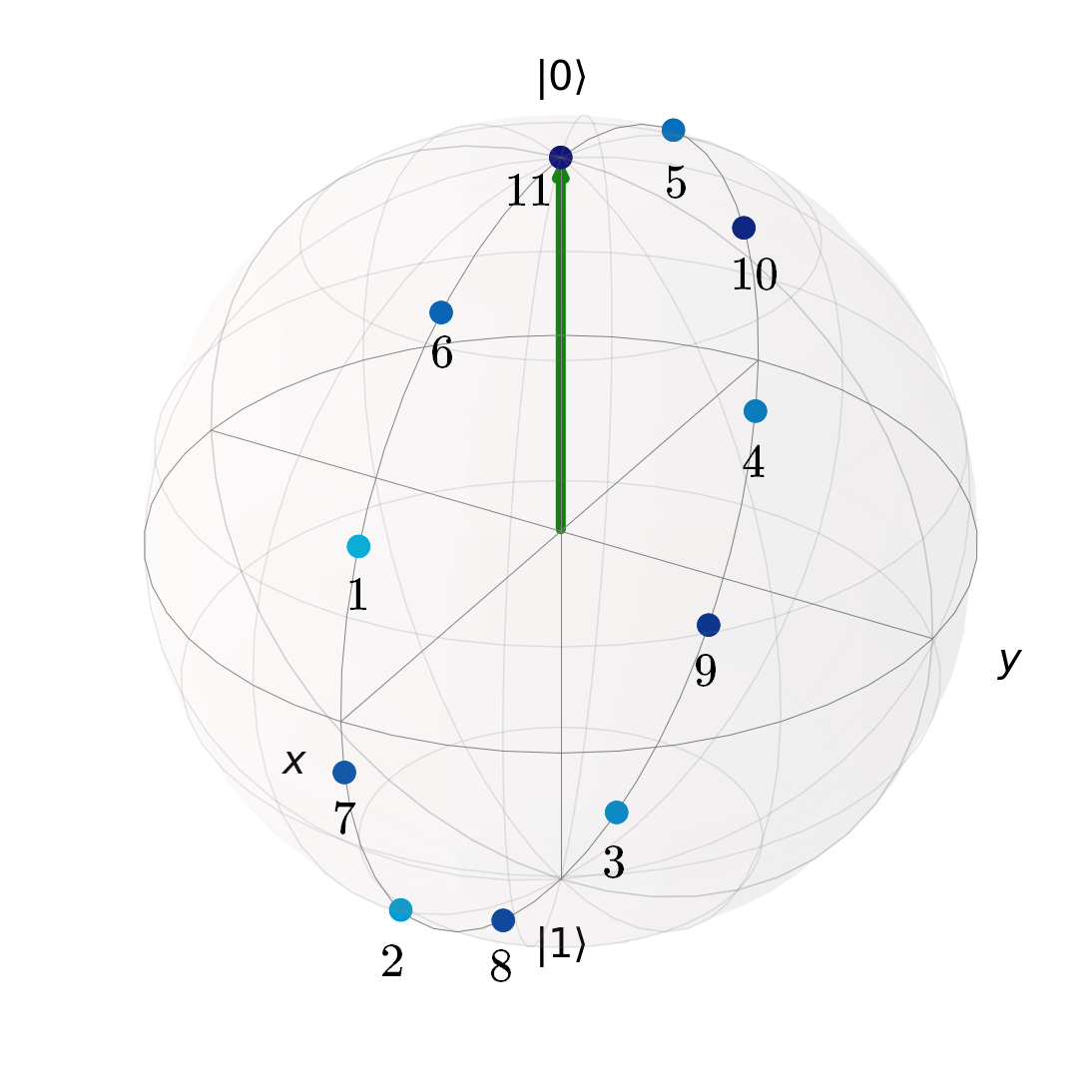}
  \caption{}
\end{subfigure}%
\begin{subfigure}{.3\textwidth}
  \centering
  \includegraphics[width=1\linewidth]{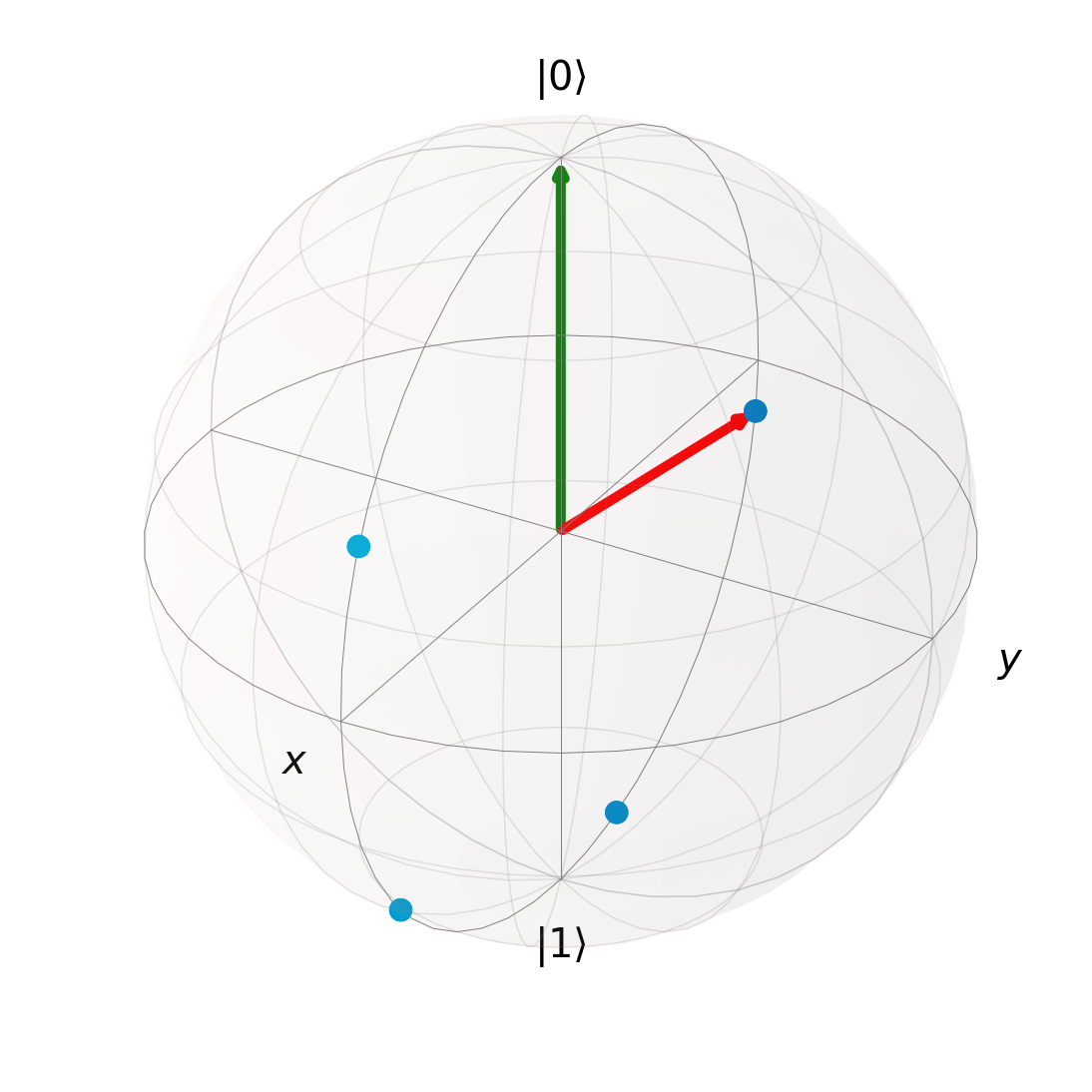}
   \caption{}
\end{subfigure}
\label{fig:test}
    \caption{Visualization of the original algorithm on the Bloch Sphere. Fig. 4(a) represents the computation for a string of length 11. The initial state is $\ket{0}$ and the quantum state of the qubit upon reading each symbol is depicted by the dots. In Fig. 4(b), the computation is visualized for a string of length 4. The final state vector is given in red.}
    \label{fig: rybloch}
\end{figure}

2-state MCQFA $\tilde{M}_{11}$ recognizing $\MOD{11}$ can be realized similarly using a single qubit. The operators $\tilde{U}_{\cent}$ and $\tilde{U}_{\$}$ are implemented using the $S_x$-gate $S_x^\dagger$-gate in Qiskit. The operator $\tilde{U}_a$ is implemented using a $R_z$-gate, providing two times the angle (see Lemma~\ref{lem: sxua}). In Figure \ref{fig: circuit1qubit}, the circuit diagram is given for the input string $aa$.

\begin{figure}[ht]
\vspace{-0.1in}
    \centering
    \begin{equation*}
        \scl{\(
        \Qcircuit @C=1.0em @R=1em{
            \lstick{q_0:\ket0}  & \gate{S_x} &\gate{R_z(\frac{4 \pi }{11})} & \gate{R_z(\frac{4 \pi }{11})} &  
            \gate{S_x^\dagger} & \meter & \rstick{c_0} \cw 
        }\)}
    \end{equation*}
    \caption{Alternative single qubit \MOD{11} implementation for the string aa. The circuit is initialized using a $S_x$-gate. For each scanned $a$, $R_z(\frac{4\pi}{11})$-gate is applied. $S_x^\dagger$-gate is applied at the end.}
    \label{fig: circuit1qubit}
\end{figure}

The computation of the new algorithm is visualized in Figure \ref{fig: rzbloch}. In (a), the initial state $\ket{0}$ is shown in green, and the state obtained after reading the right end-marker is drawn in orange. In (b), the states obtained after reading $a^j$ and applying $j$ $R_z$ gates are visualized. After reading the end-marker, the final state becomes $\ket{0}$, which is the accept state. We start in the state $\ket{0}$. After applying $ S_x $, the quantum state is mapped on the $x$-$y$ plane resulting in the following state:
\begin{equation}
    \ket{0} \xrightarrow{S_x} 
    \frac{1}{\sqrt{2}} \mypar{ \ket{0} + e^{-\frac{\pi}{2}} \ket{1}}.
\end{equation}
For each $ a $, the quantum state rotates on the $x$-$y$ plane with angle $ \frac{4\pi}{11} $ (i.e., updating only the local phase). The transition after reading the first $a$ is expressed as
\begin{equation}
    \frac{1}{\sqrt{2}} \mypar{ \ket{0} + e^{-\frac{\pi}{2}} \ket{1}}
     \xrightarrow{R_z\mypar{\frac{4\pi}{11}}  } 
    \frac{1}{\sqrt{2}} \mypar{ \ket{0} + e^{-\frac{\pi}{2}+\frac{4\pi}{11}} \ket{1}}.
\end{equation}
In general, for $ j \in \mathbf{Z}^+ $, the transition after reading the $j$\thh $a$ is given as follows:
\begin{equation}
    \frac{1}{\sqrt{2}} \mypar{ \ket{0} + e^{-\frac{\pi}{2}+(j-1)\frac{4\pi}{11}} \ket{1}}
     \xrightarrow{R_z\mypar{\frac{4\pi}{11}}  } 
    \frac{1}{\sqrt{2}} \mypar{ \ket{0} + e^{-\frac{\pi}{2}+j\frac{4\pi}{11}} \ket{1}}.
\end{equation}
After reading the 11\thh~$a$, the quantum state is expressed as in Eq.~\eqref{eq:11a}.
\begin{equation}\label{eq:11a}
    \frac{1}{\sqrt{2}} \mypar{ \ket{0} + e^{-\frac{\pi}{2}} \ket{1}}
\end{equation}
Then, after applying $ S_x^\dagger $, the quantum state returns back to $ \ket{0} $.
In (c), computation for a string of length 4 is visualized, and the final state is drawn in red. 

\begin{figure}[ht]
    \centering
\begin{subfigure}{.3\textwidth}
  \centering
  \includegraphics[width=1\linewidth]{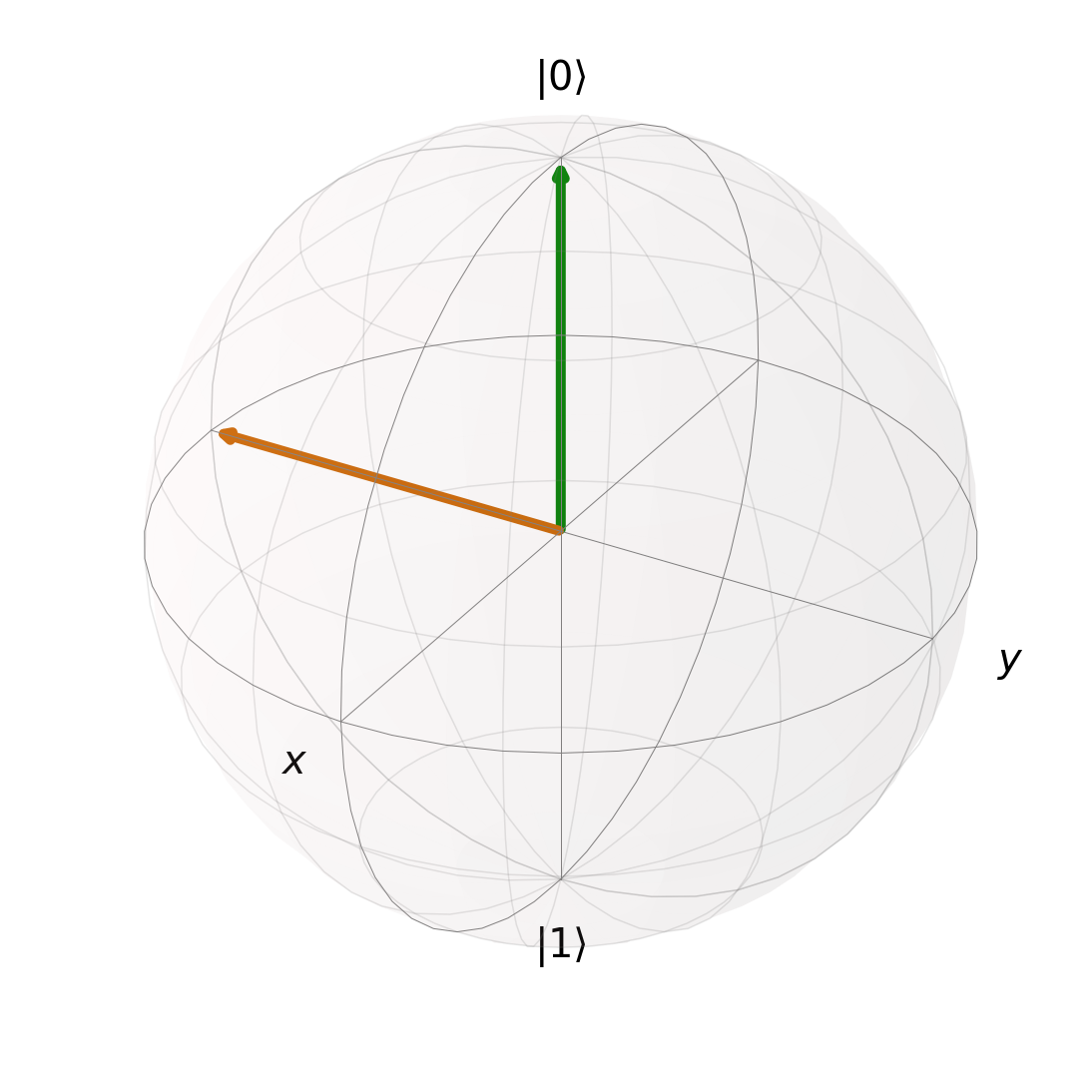}
  \caption{}
\end{subfigure}%
\begin{subfigure}{.3\textwidth}
  \centering
  \includegraphics[width=1\linewidth]{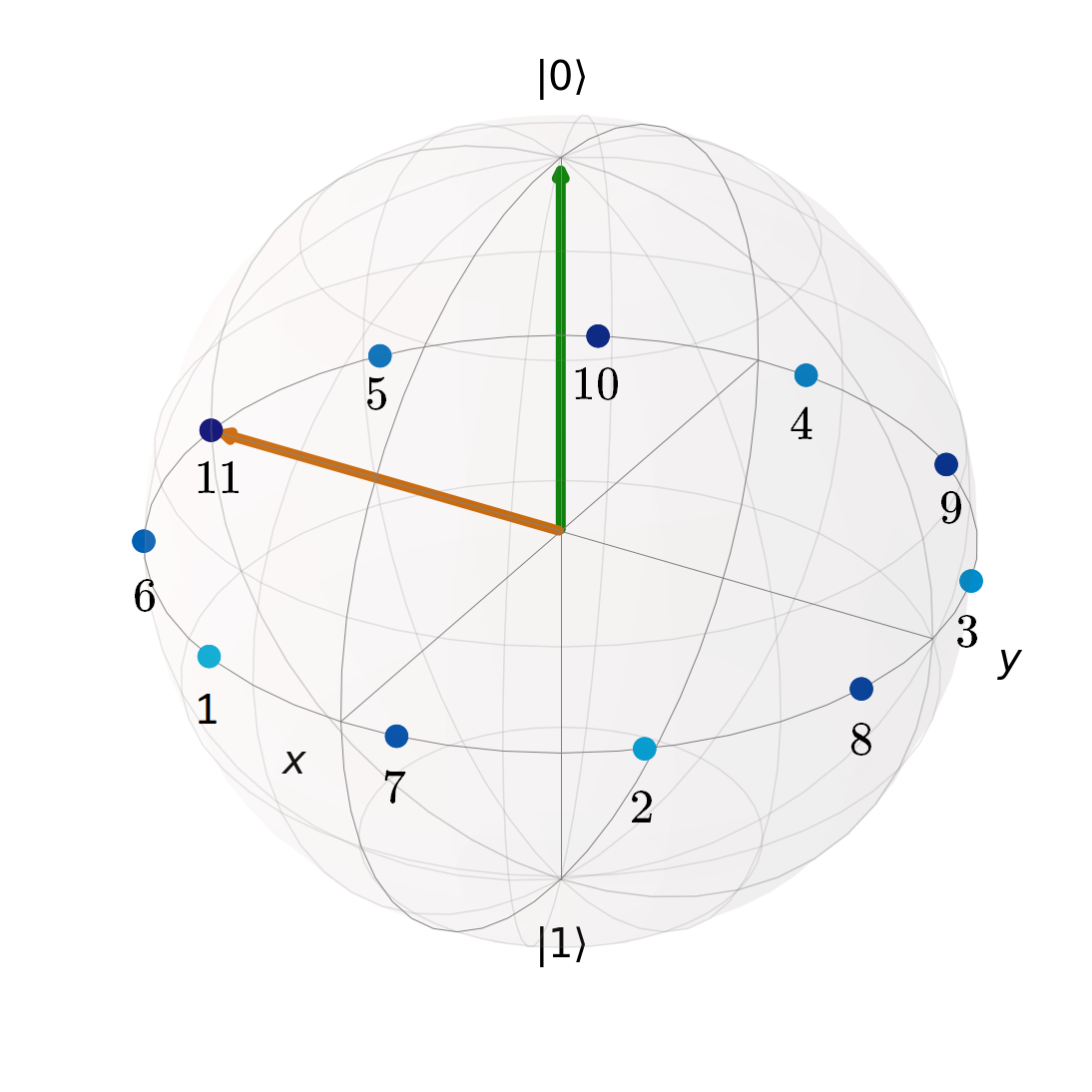}
  \caption{}
\end{subfigure}
\begin{subfigure}{.3\textwidth}
  \centering
  \includegraphics[width=1\linewidth]{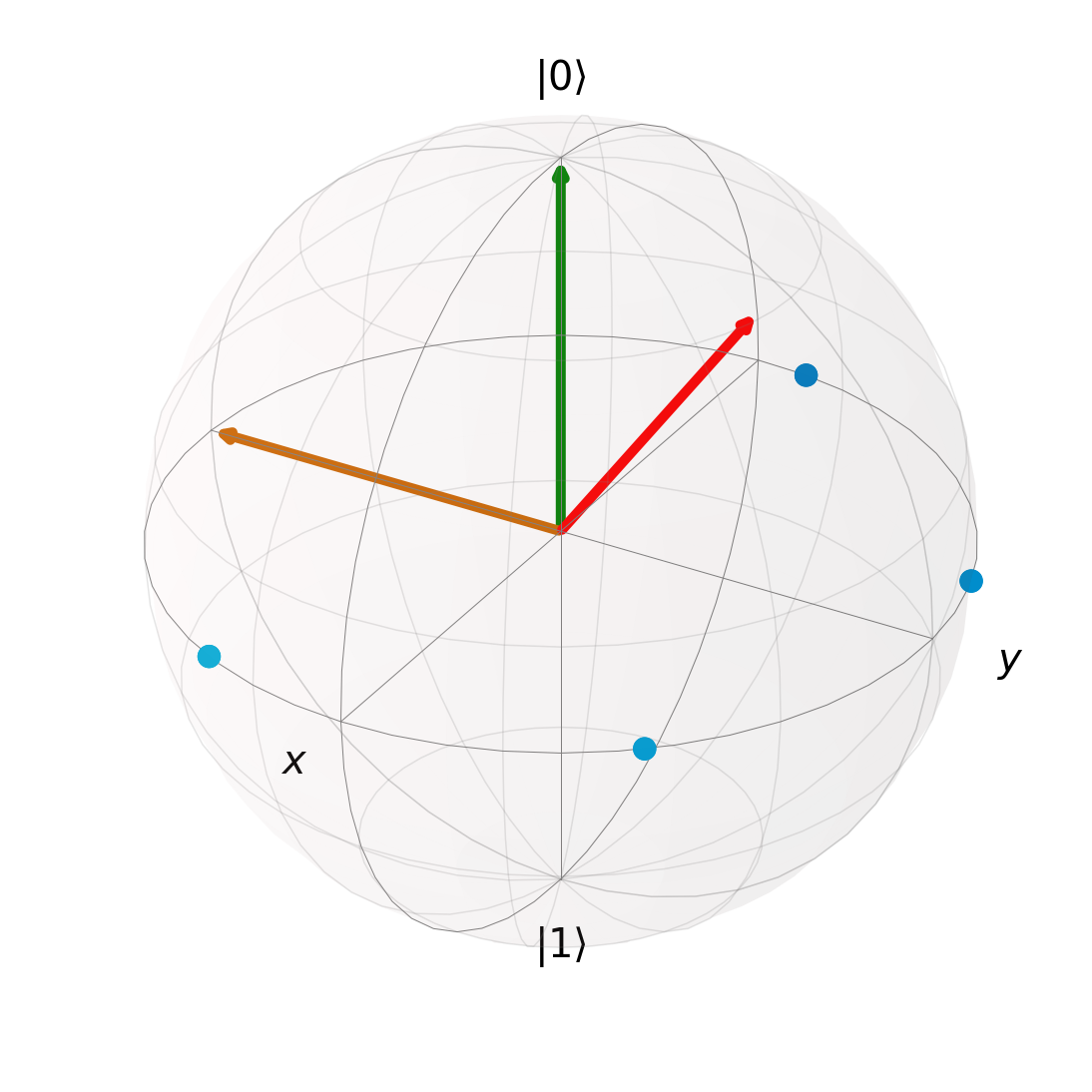}
  \caption{}
\end{subfigure}

    \caption{Visualization of the new algorithm on the Bloch Sphere. Fig. 6(a) is the initialization step, and the state obtained after applying the $S_x$ gate is given in orange. Fig. 6(b) represents the computation for a string of length 11. The quantum state of the qubit upon reading each symbol is depicted by the dots. In Fig. 6(c), the computation is visualized for a string of length 4. The final state vector is given in red.}
    \label{fig: rzbloch}
\end{figure}

Since $R_y$ gate is not among the set of basis gates, the circuit for the original algorithm is transpiled as depicted in Figure \ref{fig: transpile} before running on the real machine. Each $R_y$ gate is decomposed into two $R_z$ gates and two $S_x$ gates, overall requiring $2j$ $R_z$ gates and $2j$ $S_x$ gates for an input string $a^j$.

\begin{figure}[ht]

    \centering
    \begin{equation*}
        \scl{\(
        \Qcircuit @C=1.0em @R=1em{
         \lstick{q_0:\ket0}  &\gate{S_x}& \gate{R_z(\frac{15 \pi }{11})} &
             \gate{S_x}&
             \gate{R_z(\pi)} &
             \gate{S_x}& 
             \gate{R_z(\frac{15 \pi }{11})} &
             \gate{S_x}&
             \gate{R_z(\pi)} &
              \meter & \rstick{c_0} \cw 
        }\)}
    \end{equation*}
    \caption{Transpilation of single qubit \MOD{11} implementation for string $aa$.}
    \label{fig: transpile}
\end{figure}
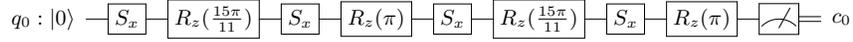

 Since the new algorithm we have proposed uses only basis gates and $S_x^\dagger$ which requires one $S_x$-gate and two $R_z$-gates when transpiled, the number of required gates is only two $S_x$ gates and $j+2$ $R_z$-gates for the input string $a^j$.

Let us note that the $R_z$ gates in Qiskit are implemented virtually, taking zero duration and with perfect accuracy \cite{MWSCG17}. Hence, running single qubit automata for any word length uses only two physical $S_x$ gates.

\subsection{Optimized Implementation}

Implementation of the operator $\tilde{U}$ is costly because it requires multiple controlled rotation gates as discussed in \cite{Kalis18,BSONY21}. To overcome this issue, implementation of the operator given in Eq.~\eqref{eq:hatua} is proposed in \cite{Kalis18} instead of implementing the unitary operator $\tilde{U}_a$.
 
  \begin{equation}\label{eq:hatua}
    \hat{U}_a = \mymatrix{cccc}{
        R_1 & 0 & 0 & 0 \\ 
        0 &  R_2 R_1 & 0 & 0 \\
        0 & 0 & R_3 R_1 & 0 \\
        0 & 0 & 0 & R_3 R_2 R_1
    } 
  \end{equation}

The proposed operator $\hat{U}_a $ does not rotate each sub-automaton with a different angle but instead each sub-automaton rotates with a combination of the chosen 3 angles. The circuit diagram implementing $\hat{U}_a $ is given in Figure \ref{circ:modp_optimized}.
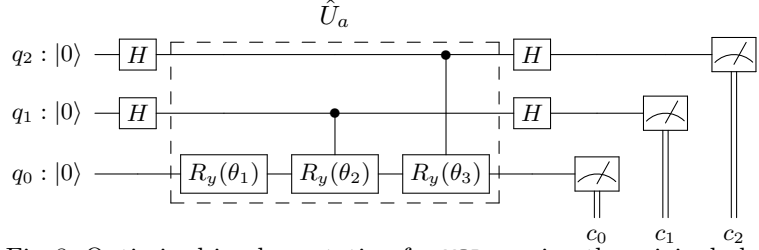
\begin{figure}[ht]
\vspace{-0.1in}
    \small
    \centering
\(
\Qcircuit @C=1em @R=1em {
    & & &\mbox{\normalsize \(\hat{U}_a\)} &  & & \\ 
 	\lstick{q_2:\ket0} & \gate{H} & \qw & \qw & \ctrl{2} &  \gate{H} & \qw & \qw & \meter\\
 	\lstick{q_1:\ket0} & \gate{H} & \qw & \ctrl{1} & \qw & \gate{H} & \qw & \meter\\
 	\lstick{q_0:\ket0} & \qw & \gate{R_y(\theta_1)} & \gate{R_y(\theta_2)} & \gate{R_y(\theta_3)} & \qw &  \meter \\
 	& & & & & & \dstick{c_0}\cwx[-1] & \dstick{c_1}\cwx[-2] & \dstick{c_2}\cwx[-3] \gategroup{2}{3}{4}{5}{.8em}{--}\\
}
\)
\caption{Optimized implementation for \MOD{11} using the original algorithm.}
    \label{circ:modp_optimized}
\end{figure}

To construct the optimized circuit for the new algorithm we have proposed, we can similarly implement the operator $\hat{U}_a'$ given in Eq.~\eqref{eq:hatuap} instead of $\tilde{U}_a'$.

\begin{equation} \label{eq:hatuap}
    \hat{U}_a' = \mymatrix{cccc}{
        R_1' & 0 & 0 & 0 \\ 
        0 &  R_2' R_1' & 0 & 0 \\
        0 & 0 & R_3' R_1' & 0 \\
        0 & 0 & 0 & R_3' R_2' R_1'
    }
\end{equation}

Let's trace the computation step by step. Upon reading the left end-marker, the machine is in the state expressed in Eq.~\eqref{eq:statecent}.
\begin{equation} \label{eq:statecent}
\ket{\psi_{\cent}} =\frac{1}{2} \sum_{j=0}^3 \ket{j}\otimes \left (\frac{1+i}{2}\ket{0} + \frac{1-i}{2}\ket{1} \right )
\end{equation}
Recalling that $R_l'$ is implemented using an $R_z$ gate say with angle $\phi_l$, after reading the first $a$ the superposition can be expressed as in Eq.~\eqref{eq:psi1}.

\begin{align} \label{eq:psi1}
    \ket{\psi_1} = 
        &\frac{1}{2} \ket{00} \otimes R_z(\phi_1)\left (\frac{1+i}{2}\ket{0} + \frac{1-i}{2}\ket{1} \right )+ \nonumber \\
        &\frac{1}{2} \ket{01} \otimes R_z(\phi_2)R_z(\phi_1)\left (\frac{1+i}{2}\ket{0} + \frac{1-i}{2}\ket{1} \right )+ \nonumber \\
        &\frac{1}{2} \ket{10} \otimes R_z(\phi_3)R_z(\phi_1)\left (\frac{1+i}{2}\ket{0} + \frac{1-i}{2}\ket{1} \right )+ \nonumber\\
        &\frac{1}{2} \ket{11} \otimes R_z(\phi_3)R_z(\phi_2)R_z(\phi_1)\left (\frac{1+i}{2}\ket{0} + \frac{1-i}{2}\ket{1} \right )
\end{align}
Letting $\theta_1 = \phi_1$, $\theta_2 = \phi_2 + \phi_1$ and $\theta_3 = \phi_1 + \phi_2 + \phi_3$, we
give the circuit diagram in Figure \ref{circ:modp_optimized_rz}.

\begin{figure}[ht]
\vspace{-0.1in}
    \small
    \centering
    \(
\Qcircuit @C=1em @R=1em {
    & & &\mbox{\normalsize \(\hat{U}_a'\)} &  & & \\ 
 	\lstick{q_2:\ket0} & \gate{H} & \qw & \qw & \ctrl{2} &  \gate{H} & \qw & \qw & \meter\\
 	\lstick{q_1:\ket0} & \gate{H} & \qw & \ctrl{1} & \qw & \gate{H} & \qw & \meter\\
 	\lstick{q_0:\ket0} & \gate{S_X} & \gate{R_z(\theta_1)} & \gate{R_z(\theta_2)} & \gate{R_z(\theta_3)} & \gate{S_X^{\dagger}} &  \meter \\
 	& & & & & & \dstick{c_0}\cwx[-1] & \dstick{c_1}\cwx[-2] & \dstick{c_2}\cwx[-3] \gategroup{2}{3}{4}{5}{.8em}{--}\\
}
\)
    \caption{Optimized implementation for \MOD{11} using the new algorithm.}
    \label{circ:modp_optimized_rz}
\end{figure}
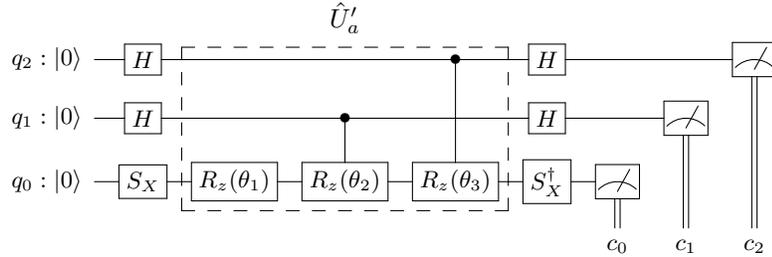

After the transpilation, the number of basis gates required by both implementations and the depth of the circuit for a string of length 11 is given in Table \ref{tabl:gates}. The depth of the circuit is the longest path from the input to the output of the circuit and is calculated using Qiskit's built-in function. Note that the transpilation of controlled rotation gates introduces $CX$ gates. The difference in the number of $S_x$ and $R_z$ gates is due to different transpilation schemes for controlled $R_y$ and $R_z$ gates. 

\begin{table}[h]
    \centering
    \begin{tabular}{lllll}
        \toprule
         & $S_X$ & $R_Z$ & $CX$ & Depth \\
         \midrule
        Original  & 70    & 74    & 44 & 183    \\
        New  & 6    & 44    & 44 &  84 \\
        \bottomrule
        \\
    \end{tabular}
    \caption{The number of basis gates required by the two implementations.}
    \label{tabl:gates}
\end{table}
We ran both implementations on $\mathit{ibmq\_belem}$ letting $K=\{3,5,7\}$. The experiments are repeated 3 times, and the number of shots is set to 8192. The average acceptance probabilities for each word length are plotted in Figure \ref{fig: plot}. The black lines above the bars show the minimum and maximum acceptance probabilities. Ideal probability shows the acceptance probability that is calculated analytically.

\begin{figure}
    \centering
    \includegraphics[width= 0.9 \textwidth]{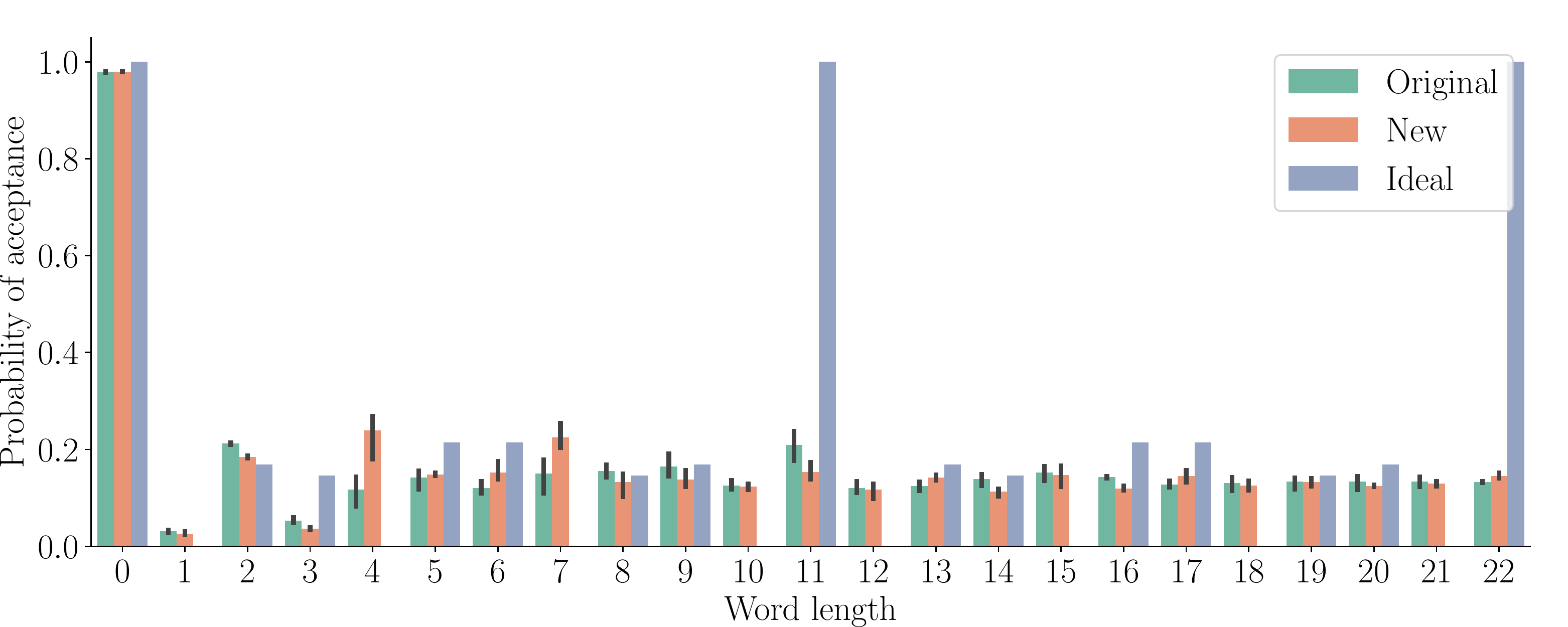}
    \caption{Acceptance probabilities for various word lengths.}
    \label{fig: plot}
\end{figure}

Ideally, the accepting probabilities for lengths 0, 11, and 22 are 1, and the accepting probabilities for the others are bounded above by 0.22. As seen from the figure, the experimental results are far from the ideal for lengths 11 and 22, and for the rest, both the experimental results and ideal results are bounded above. When considering both cases, we can say that the experimental results are quickly getting random as their statistics are similar to each other, even for lengths 11 and 22. 

The main reason behind the above behavior is that the available real quantum computers are heavily prone to noise \cite{NISQ,IBM,CKJMC19}. There are various sources of noise such as gate errors, measurement readout errors, and thermal relaxation errors \cite{WO20}. Various experiments are run during the calibration of the devices, to characterize the noise and determine the error parameters for each backend. The estimated error parameters for each backend are available through Qiskit as the \textit{basic device noise model}, although with some simplifications \cite{WO20}.  The basic device noise model allows performing noisy simulations with the noise parameters of the real machines. Thus, we can go one step ahead and obtain further evidence for our experimental results. 

We execute both implementations on the simulators mimicking the noise of real machines, and then we check how close the ideal quantum states and noisy-simulated quantum states are. We use the fidelity for measuring the closeness of the states, which is defined as in Eq.~\eqref{eq:fid},
\begin{equation}\label{eq:fid}
 F(\sigma,\rho) =   \left (tr \sqrt{\sqrt{\rho} \sigma \sqrt{\rho} }\right )^2
\end{equation}
where $\rho$ and $\sigma$ are two density matrices. In our case, the ideal state is a pure state while the state obtained as a result of the noisy simulation is a mixed state; the formula reduces to $\bra{\psi} \sigma \ket{\psi}$ where $\rho = \ket{\psi}\bra{\psi}$. We calculate the fidelity of the states after reading the inputs $ a, aa, aaa, \ldots, a^{22} $. The results are presented in Table~\ref{tabl: gates} and visualized in Fig.~\ref{fig: fidelity}.

\begin{table}[h]
 \setlength{\tabcolsep}{5pt}
\centering
\begin{tabular}{|c|c|c|c|c|c|}
\hline
\textbf{Length} & \textbf{Original} &  \textbf{New}  & \textbf{Length} & \textbf{Original} & \textbf{New}  \\ \hline
\textbf{1}  & 0.862 & 0.880 & \textbf{12} & 0.405 & 0.450 \\ \hline
\textbf{2}  & 0.727 & 0.830 & \textbf{13} & 0.412 & 0.410 \\ \hline
\textbf{3}  & 0.766 & 0.750 & \textbf{14} & 0.380 & 0.380 \\ \hline
\textbf{4}  & 0.717 & 0.720 & \textbf{15} & 0.344 & 0.370 \\ \hline
\textbf{5}  & 0.612 & 0.680 & \textbf{16} & 0.342 & 0.360 \\ \hline
\textbf{6}  & 0.621 & 0.630 & \textbf{17} & 0.309 & 0.320 \\ \hline
\textbf{7}  & 0.477 & 0.590 & \textbf{18} & 0.315 & 0.310 \\ \hline
\textbf{8}  & 0.506 & 0.560 & \textbf{19} & 0.301 & 0.310 \\ \hline
\textbf{9}  & 0.514 & 0.510 & \textbf{20} & 0.287 & 0.290 \\ \hline
\textbf{10} & 0.481 & 0.500 & \textbf{21} & 0.258 & 0.280 \\ \hline
\textbf{11} & 0.459 & 0.470 & \textbf{22} & 0.266 & 0.270 \\ \hline
\end{tabular}%
    \caption{The fidelity of ideal quantum states and noisy-simulated states calculated for both experiments.}
    \label{tabl: gates}
\end{table}

\begin{figure}[h!]
\begin{center}
\includegraphics[scale=0.7]{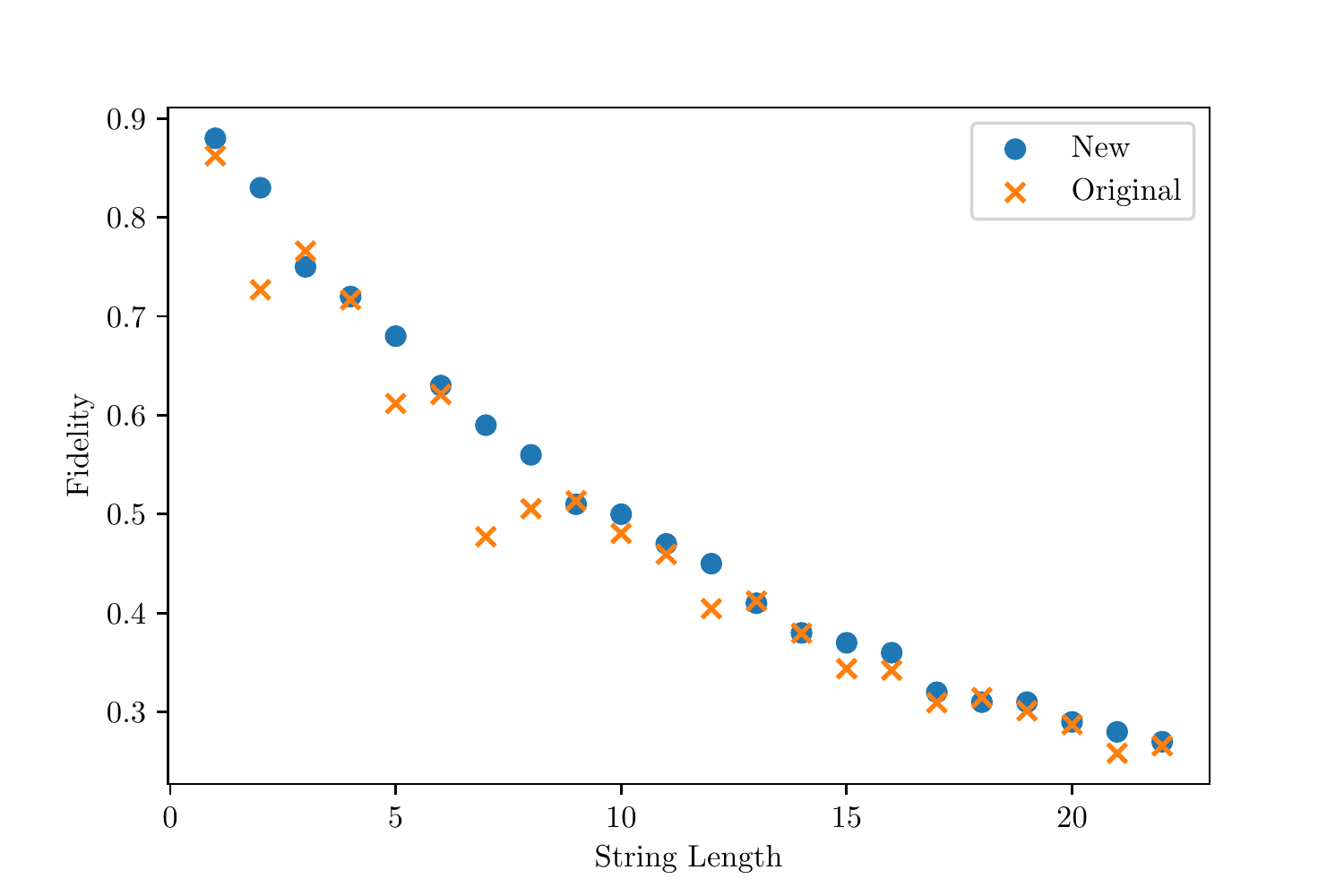}
\caption{The fidelity of ideal quantum states and noisy-simulated states calculated for both experiments is depicted in the graph.}
\label{fig: fidelity}
\end{center}
\end{figure}

If there were no noise or negligible noise, the fidelity would be 1 or close to 1, respectively. As seen from the data, even after reading a single symbol, the fidelity is away from 1, and, for the longer inputs, the fidelity quickly drops. We observe the accumulation of noise, as the fidelity continues to decrease with the length of the input.

Although the new implementation requires fewer gates and the circuit has a smaller depth, unfortunately, the experiment results do not look promising due to noise.

\section{Conclusion and Future Work}

In this paper, we presented a modified MCQFA algorithm for the $\MODp$ problem, whose circuit implementation has smaller depth and uses fewer gates compared to the original algorithm when run on the IBMQ backends. Despite the unsatisfying results from the real backend, the study contributes both to the field of QFA and the implementation of quantum algorithms on real devices. Further studies should be conducted to mitigate the effect of noise in the MCQFA setting, which is not a trivial task since operators should be applied one at a time by definition of MCQFA. 

\section*{Acknowledgements}
This work has been supported by the Polish National Science Center under the grant agreement 2019/33/B/ST6/02011. Yakary{\i}lmaz was partially supported by the ERDF project Nr. 1.1.1.5/19/A/005 ``Quantum computers with constant memory'' and the project ``Quantum algorithms: from complexity theory to experiment'' funded under ERDF programme 1.1.1.5.

We thank Utku Birkan, Cem Nurlu, and Viktor Olejar for certain discussions on the implementations of QFA algorithms by using operators $SX$ and $ R_z(\theta) $. We thank Robert L. Svarinskis for his corrections.

We acknowledge the use of IBM Quantum services for this work. We acknowledge the access to advanced services provided by the IBM Quantum Researchers Program.

The code used for the experiments is available at \cite{code}.

\bibliographystyle{splncs04}

\bibliography{references}

\begin{thebibliography}{10}
\providecommand{\url}[1]{\texttt{#1}}
\providecommand{\urlprefix}{URL }
\providecommand{\doi}[1]{https://doi.org/#1}

\bibitem{IBM}
{IBM Q}uantum (2021), \url{https://quantum-computing.ibm.com/}

\bibitem{AF98}
Ambainis, A., Freivalds, R.: 1-way quantum finite automata: Strengths,
  weaknesses and generalizations. In: 39th Annual Symposium on Foundations of
  Computer Science, {FOCS} '98. pp. 332--341. {IEEE} Computer Society (1998)

\bibitem{AN09}
Ambainis, A., Nahimovs, N.: Improved constructions of quantum automata.
  Theoretical Computer Science  \textbf{410}(20),  1916--1922 (2009)

\bibitem{AY12}
Ambainis, A., Yakary{\i}lmaz, A.: Superiority of exact quantum automata for
  promise problems. Information Processing Letters  \textbf{112}(7),  289--291
  (2012)

\bibitem{AY21}
Ambainis, A., Yakary{\i}lmaz, A.: Automata and quantum computing. In: Éric
  Pin, J. (ed.) Handbook of Automata Theory, vol.~2, chap.~39, pp. 1457--1493.
  European Mathematical Society Publishing House (2021)

\bibitem{BSONY21}
Birkan, U., Salehi, {\"O}., Olejar, V., Nurlu, C., Yakary{\i}lmaz, A.:
  Implementing quantum finite automata algorithms on noisy devices. In:
  International Conference on Computational Science. pp. 3--16. Springer (2021)

\bibitem{CKJMC19}
C{\'o}rcoles, A.D., Kandala, A., Javadi-Abhari, A., McClure, D.T., Cross, A.W.,
  Temme, K., Nation, P.D., Steffen, M., Gambetta, J.M.: Challenges and
  opportunities of near-term quantum computing systems. arXiv preprint
  arXiv:1910.02894  (2019)

\bibitem{Kalis18}
K\={a}lis, M.: Kvantu Algoritmu Realiz\={a}cija Fizisk\={a} Kvantu Dator\={a}.
  Master's thesis, University of Latvia (2018)

\bibitem{MWSCG17}
McKay, D.C., Wood, C.J., Sheldon, S., Chow, J.M., Gambetta, J.M.: Efficient
  {$Z$} gates for quantum computing. Physical Review A  \textbf{96}(2),  022330
  (2017)

\bibitem{MPC20}
Mereghetti, C., Palano, B., Cialdi, S., Vento, V., Paris, M.G.A., Olivares, S.:
  Photonic realization of a quantum finite automaton. Phys. Rev. Research
  \textbf{2}(1),  013089--013103 (2020)

\bibitem{MC00}
Moore, C., Crutchfield, J.P.: Quantum automata and quantum grammars.
  Theoretical Computer Science  \textbf{237}(1-2),  275--306 (2000)

\bibitem{PHYF22}
Plachta, S.Z.D., Hiekkam{\"{a}}ki, M., Yakary{\i}lmaz, A., Fickler, R.: Quantum
  advantage using high-dimensional twisted photons as quantum finite automata.
  Quantum  \textbf{6}, ~752 (2022). \doi{10.22331/q-2022-06-30-752},
  \url{https://doi.org/10.22331/q-2022-06-30-752}

\bibitem{NISQ}
Preskill, J.: Quantum computing in the {NISQ} era and beyond. Quantum
  \textbf{2},  79--590 (2018)

\bibitem{code}
Salehi, {\"O}., Yakaryılmaz, A.: iitis/mcqfa: v1.1,
  \url{https://doi.org/10.5281/zenodo.6023333}

\bibitem{TFL19}
Tian, Y., Feng, T., Luo, M., Zheng, S., Zhou1, X.: Experimental demonstration
  of quantum finite automaton. npj Quantum Inf  \textbf{5}(56) (2019)

\bibitem{WO20}
Wood, C.J.: Special session: Noise characterization and error mitigation in
  near-term quantum computers. In: 2020 IEEE 38th International Conference on
  Computer Design (ICCD). pp. 13--16. IEEE (2020)

\end{thebibliography}

\end{document}